\documentclass[journal,onecolumn]{IEEEtran}

\usepackage{mathrsfs}
\usepackage{color}
\usepackage{enumerate}
\usepackage{mathrsfs}
\usepackage{amssymb}
\usepackage{amsfonts}
\usepackage{amsmath}
\usepackage{multirow}
\usepackage{amsthm}
\usepackage{url}

\newtheorem{theorem}{Theorem}[section]
\newtheorem{definition}{Definition}[section]
\newtheorem{lemma}[theorem]{Lemma}
\newtheorem{example}[theorem]{Example}

\newtheorem{proposition}[theorem]{Proposition}

\newtheorem{remark}{Remark}[section]

\begin{document}

\title{Quantum Block and Synchronizable Codes Derived from Certain Classes of Polynomials}

\author{Tao Zhang and Gennian Ge
\thanks{The research of G. Ge was supported by the National Natural Science Foundation of China under Grant No.~61171198 and  Grant No.~11431003, the Importation and Development of High-Caliber Talents Project of Beijing Municipal Institutions, and Zhejiang Provincial Natural Science Foundation of China under Grant No.~LZ13A010001.}
\thanks{T. Zhang is with  the School of Mathematical Sciences, Capital Normal University,
Beijing 100048, China. He is also with the Department of Mathematics, Zhejiang University,
Hangzhou 310027,  China (e-mail: tzh@zju.edu.cn).}
\thanks{G. Ge is with  the School of Mathematical Sciences, Capital Normal University,
Beijing 100048, China. He is also with Beijing Center for Mathematics and Information Interdisciplinary Sciences, Beijing, 100048, China (e-mail: gnge@zju.edu.cn).}
}

\maketitle

\begin{abstract}
One central theme in quantum error-correction is to construct quantum codes that have a large minimum distance. In this paper, we first present a construction of classical codes based on certain class of polynomials. Through these classical codes, we are able to obtain some new quantum codes. It turns out that some of quantum codes exhibited here have better parameters than the ones available in the literature. Meanwhile, we give a new class of quantum synchronizable codes with highest possible tolerance against misalignment from duadic codes.
\end{abstract}

\section{Introduction}
Quantum codes were introduced to protect quantum information from decoherence during quantum computations and quantum communications. After the initial work for quantum error-correcting codes \cite{S95,S96}, researchers have made great progress in developing quantum codes. In \cite{CRSS98}, the construction of binary quantum codes was diverted into finding classical self-orthogonal codes over $\textup{GF}(4)$, and then generalized to the nonbinary case \cite{R99,AK01}. After the establishment of the above connection between quantum codes and classical codes, the construction of quantum codes can be converted to that of classical self-orthogonal codes (see \cite{AKS07,CLX05,DJX12,FX08,JX14,G14,LXW08,L04,S99} and the references therein). As in classical coding theory, one of the central tasks in quantum coding theory is to construct quantum code with a large minimum distance.

 We know that Reed-Solomon code is an MDS (maximum distance separable) code \cite{HP03}, as its parameters meet the Singleton bound. But the length of a $q$-ary Reed-Solomon code is restricted to be less than or equal to $q$. By using the principal of evaluation, several successful constructions have been made to construct good $q$-ary linear codes which have lengths larger than $q$ (see \cite{XL00,LNX01,JX12,SH14}). Motivated by these works, we will give a new construction of linear codes by evaluation. Using these linear codes, we obtain new quantum codes. Some of them have better parameters than the quantum codes listed in tables online \cite{E12}, \cite{G07}.

Recently, Fujiwara \cite{F13} considered another type of quantum codes, named quantum synchronizable codes, which is a coding scheme that corrects not only the general quantum noise represented by Pauli errors but also misalignment errors. Misalignment is the type of errors that the information processing device misidentifies the boundaries of an information block. For example, assume that each chunk of information is encoded into a block of consecutive three bits in a stream of bits $b_{i}$ so that the data has a frame structure. If three blocks of information are encoded, we have nine ordered bits $(b_{0},b_{1},b_{2},b_{3},b_{4},b_{5},b_{6},b_{7},b_{8})$ in which each of the three blocks $(b_{0},b_{1},b_{2})$, $(b_{3},b_{4},b_{5})$, $(b_{6},b_{7},b_{8})$ forms an information chunk. If misalignment occurs to the right by one bit when attempting to retrieve the second block of information, the device will wrongly read out $b_{4}$, $b_{5}$ and $b_{6}$ instead of the correct set of bits $b_{3}$, $b_{4}$ and $b_{5}$. The same kind of error in block synchronization may be considered for a stream of qubits.

A theoretical framework of quantum synchronizable coding was first introduced in \cite{F13}. Later, the authors of \cite{FTW13} improved the original construction method and gave more examples of quantum synchronizable codes. While we have a theoretical framework of quantum synchronizable coding, there are only a few infinite families of quantum synchronizable codes in the literature. One of the main challenges in the construction of quantum synchronizable codes is that it is quite difficult to find suitable classical codes because the required algebraic structures are very harsh. Particular constraint is the variety of available code parameters. For instance, the lengths of the encoded information blocks of the known quantum synchronizable codes that have highest possible tolerance against misalignment are primes, $2^{m}-1$ or of the form $\frac{2^{hm}-1}{2^{h}-1}$ \cite{FTW13,FV14,XYF14}. In this work, we construct a new infinite family of quantum synchronizable codes by exploiting duadic codes over finite field $\mathbb{F}_{2}$.

This paper is organized as follows. In Section~\ref{pre} we recall the basics about quantum codes and quantum synchronizable codes. In Section~\ref{quantumcodes}, we first give a new construction of classical linear codes. From these linear codes, we obtain some new quantum codes. In Section~\ref{quantumsyn}, we present a new class of quantum synchronizable codes based on duadic codes.

\section{Preliminaries}\label{pre}
Let $\mathbb{F}_{q}$ be the finite field with $q$ elements, where $q$ is a prime power. A linear $[n,k]$ code $C$ over $\mathbb{F}_{q}$ is a $k$-dimensional subspace of $\mathbb{F}_{q}^{n}$. The weight $\textup{wt}(x)$ of a codeword $x\in C$ is the number of nonzero components of $x$. The distance of two codewords $x,y\in C$ is $d(x,y):=\textup{wt}(x-y)$. The minimum distance $d$ between any two distinct codewords of $C$ is called the minimum distance of $C$. An $[n,k,d]$ code is an $[n,k]$ code with the minimum distance $d$.

A linear $[n,k]$ code $C$ over the finite field $\mathbb{F}_{q}$ is called cyclic if $(c_{0},c_{1},\cdots,c_{n-1})\in C$ implies $(c_{n-1},c_{0},c_{1},\cdots,c_{n-2})\in C$. If we identify each codeword $(c_{0},c_{1},\cdots,c_{n-1})$ with $c_{0}+c_{1}x+c_{2}x^{2}+\cdots+c_{n-1}x^{n-1}\in\mathbb{F}_{q}[x]/(x^{n}-1)$, then cyclic code $C$ is identified with an ideal of the ring $\mathbb{F}_{q}[x]/(x^{n}-1)$. Note that every ideal of $\mathbb{F}_{q}[x]/(x^{n}-1)$ is principal, then $C$ is generated by a monic divisor $g(x)$ of $x^{n}-1$. In this case, $g(x)$ is called the generator polynomial of $C$ and we write $C=\langle g(x)\rangle$. Let $\alpha$ be a primitive $n$-th root of unity in some extension field of $\mathbb{F}_{q}$, then the set $Z=\{i|g(\alpha^{i})=0,\ 0\leq i\leq n-1\}$ is called the defining set of $C$.

Given two vectors $x=(x_{0},x_{1},\cdots,x_{n-1}),\ y=(c_{0},c_{1},\cdots,c_{n-1})\in\mathbb{F}_{q}^{n}$, there are two inner products we are interested in. One is the Euclidean inner product which is defined as
$$\langle x,y\rangle_{E}=x_{0}y_{0}+x_{1}y_{1}+\cdots+x_{n-1}y_{n-1}.$$
When $q=l^{2}$, where $l$ is a prime power, then we can also consider the Hermitian inner product which is defined by
$$\langle x,y\rangle_{H}=x_{0}y_{0}^{l}+x_{1}y_{1}^{l}+\cdots+x_{n-1}y_{n-1}^{l}.$$
The Euclidean dual code of $C$ is defined as
$$C^{\bot E}=\{x\in\mathbb{F}_{q}^{n}|\langle x,y\rangle_{E}=0\textup{ for all }y\in C\}.$$
Similarly the Hermitian dual code of $C$ is defined as
$$C^{\bot H}=\{x\in\mathbb{F}_{q}^{n}|\langle x,y\rangle_{H}=0\textup{ for all }y\in C\}.$$
A linear code $C$ is called Euclidean (Hermitian) dual-containing if $C^{\bot E}\subseteq C$ ($C^{\bot H}\subseteq C$, respectively). A code $D$ is called $C$-containing if $C\subseteq D$.
\subsection{Quantum Codes}
In this subsection, we recall the basics of quantum codes. Let $q$ be a power of a prime number $p$. A qubit $|v\rangle$ is a nonzero vector in $\mathbb{C}^{q}$: $|v\rangle=\sum_{x\in\mathbb{F}_{q}}c_{x}|x\rangle$, where $\{|x\rangle|x\in\mathbb{F}_{q}\}$ is a basis of $\mathbb{C}^{q}$. For $n\geq1$, the $n$-th tensor product $(\mathbb{C}^{q})^{\bigotimes n}\cong\mathbb{C}^{q^{n}}$ has a basis $\{|a_{1}\cdots a_{n}\rangle=|a_{1}\rangle\bigotimes\cdots\bigotimes|a_{n}\rangle|(a_{1},\cdots,a_{n})\in\mathbb{F}_{q}^{n}\},$ then an $n$-qubit is a nonzero vector in $\mathbb{C}^{q^{n}}$: $|v\rangle=\sum_{a\in\mathbb{F}_{q}^{n}}c_{a}|a\rangle,$ where $c_{a}\in \mathbb{C}.$

Let $\zeta_{p}$ be a primitive $p$-th root of unity. The quantum errors in $q$-ary quantum system are linear operators acting on $\mathbb{C}^{q}$ and can be represented by the set of error bases: $\varepsilon_{n}=\{T^{a}R^{b}|a,b\in\mathbb{F}_{q}\}$, where $T^{a}R^{b}$ is defined by
$$T^{a}R^{b}|x\rangle=\zeta_{p}^{\textup{Tr}_{\mathbb{F}_{q}/\mathbb{F}_{p}}(bx)}|x+a\rangle.$$
The set
$$E_{n}=\{\zeta_{p}^{l}T^{a}R^{b}|0\leq l\leq p-1,a=(a_{1},\cdots,a_{n}),b=(b_{1},\cdots,b_{n})\in\mathbb{F}_{q}^{n}\}$$
forms an error group, where $\zeta_{p}^{l}T^{a}R^{b}$ is defined by
$$\zeta_{p}^{l}T^{a}R^{b}|x\rangle=\zeta_{p}^{l}T^{a_{1}}R^{b_{1}}|x_{1}\rangle\bigotimes\cdots\bigotimes T^{a_{n}}R^{b_{n}}|x_{n}\rangle=\zeta_{p}^{l+\textup{Tr}_{\mathbb{F}_{q}/\mathbb{F}_{p}}(bx)}|x+a\rangle,$$
for any $|x\rangle=|x_{1}\rangle\bigotimes\cdots\bigotimes|x_{n}\rangle$, $x=(x_{1},\cdots,x_{n})\in\mathbb{F}_{q}^{n}$. For an error $e=\zeta_{p}^{l}T^{a}R^{b}$, its quantum weight is defined by
$$w_{Q}(e)=\sharp\{1\leq i\leq n|(a_{i},b_{i})\neq(0,0)\}.$$

A subspace $Q$ of $\mathbb{C}^{q^{n}}$ is called a $q$-ary quantum code with length $n$. The $q$-ary quantum code has minimum distance $d$ if and only if it can detect all errors in $E_{n}$ of quantum weight less than $d$, but cannot detect some errors of weight $d$. A $q$-ary $[[n,k,d]]_{q}$ quantum code is a $q^{k}$-dimensional subspace of $\mathbb{C}^{q^{n}}$ with minimum distance $d$. There are many methods to construct quantum codes, and the following theorem is one of the most frequently used construction methods.
\begin{theorem}$($\rm{\cite{AK01}} Hermitian Construction$)$\label{thm1}
If $C$ is a $q^{2}$-ary Hermitian dual-containing $[n,k,d]$ code, then there exists a $q$-ary $[[n,2k-n,\geq d]]$-quantum code.
\end{theorem}

We also mention that quantum codes have propagation rules as classical linear codes.
\begin{theorem}$($\rm{\cite{FLX06}} Propagation Rules$)$ \label{propa}
Suppose there is a $q$-ary $[[n,k,d]]$ quantum code. Then
\begin{enumerate}
  \item (Subcode) there exists a $q$-ary $[[n,k-1,\geq d]]$ quantum code.
  \item (Lengthening) there exists a $q$-ary $[[n+1,k,\geq d]]$ quantum code.
  \item (Puncturing) there exists a $q$-ary $[[n-1,k,\geq d-1]]$ quantum code.
  \item there exists a $q$-ary $[[n,k,d-1]]$ quantum code.
\end{enumerate}
\end{theorem}
\subsection{Quantum Synchronizable Codes}
 A $(c_{l},c_{r})-[[n,k]]_{2}$ quantum synchronizable code is an $[[n,k]]_{2}$ quantum code that corrects not only bit errors and phase errors but also misalignment to the left by $c_{l}$ qubits and to the right by $c_{r}$ qubits for some nonnegative integers $c_{l}$ and $c_{r}$. The framework for constructing quantum synchronizable codes involves an algebraic notion in polynomial rings. Let $f(x)\in\mathbb{F}_{2}[x]$ be a polynomial over $\mathbb{F}_{2}$ such that $f(0)=1$. The cardinality $\textup{ord}(f(x))=|\{x^{a}\pmod{f(x)}|a\in\mathbb{N}\}|$ is called the order of the polynomial $f(x)$, where $\mathbb{N}$ is the set of positive integers.

\begin{theorem}\rm{\cite{FTW13}}\label{syn}
Let $C$ be a Euclidean dual-containing cyclic $[n,k_{1},d_{1}]_{2}$ code with generator polynomial $h(x)$ and $D$ be a $C$-containing cyclic $[n,k_{2},d_{2}]_{2}$ code with generator polynomial $g(x)$. Define polynomial $f(x)$ of degree $k_{2}-k_{1}$ to be the quotient of $h(x)$ divided by $g(x)$ over $\mathbb{F}_{2}[x]/(x^{n}-1)$. Then for every pair $a_{l},a_{r}$ of nonnegative integers such that $a_{l}+a_{r}<\textup{ord}(f(x))$ there exists a quantum synchronizable $(a_{l},a_{r})-[[n+a_{l}+a_{r},2k_{1}-n]]_{2}$ code that corrects at least up to $\lfloor\frac{d_{1}-1}{2}\rfloor$ phase errors and at least up to $\lfloor\frac{d_{2}-1}{2}\rfloor$ bit errors.
\end{theorem}

In order to compute the order of the polynomial $f(x)$, we need the following lemmas.

\begin{lemma}\rm{\cite{LN97}}\label{ord1}
Let $f(x)\in\mathbb{F}_{q}[x]$ be an irreducible polynomial over $\mathbb{F}_{q}$ of degree $m$ and with $f(0)\neq0$. Then $\textup{ord}(f(x))$ is equal to the order of any root of $f$ in the multiplicative group $\mathbb{F}_{q^{m}}^{*}$.
\end{lemma}

\begin{lemma}\rm{\cite{LN97}}\label{ord2}
Let $g_{1}(x),\cdots,g_{k}(x)$ be pairwise relatively prime nonzero polynomials over $\mathbb{F}_{q}$, and let $f(x)=g_{1}(x)\cdots g_{k}(x)$. Then $\textup{ord}(f(x))$ is equal to the least common multiple of $\textup{ord}(g_{1}(x)),\cdots,\textup{ord}(g_{k}(x))$.
\end{lemma}
\section{New Quantum Codes from Polynomial Codes}\label{quantumcodes}
This section gives a new construction of classical linear codes and provides a family of quantum codes. The following notations are fixed in this section.
\begin{itemize}
  \item Let $q=p^{a}$ be a prime power, where $p$ is a prime number.
  \item Let $m$ be a positive integer with $\textup{gcd}(q,m)=1$ and $\textup{ord}_{m}(q)=p^{b}$, where $b\geq1$ is a positive integer.
  \item Let $t=p^{b}$ and $s=p^{b-1}$.
\end{itemize}
\subsection{Polynomial Codes}
Note that $\textup{ord}_{m}(q)=t$. For any $a\in\mathbb{Z}_{m}=\{0,1,\cdots,m-1\}$, the $q$-cyclotomic coset $C_{a}$ belonging to $a$ is defined by
$$C_{a}:=\{aq^{j}\pmod{m}|0\leq j\leq t-1\}.$$
Now we choose a representation of the $q$-cyclotomic coset with maximum length:
$$A:=\{\textup{max}(C_{a})|0\leq a\leq m-1, |C_{a}|=t\},$$
where $\textup{max}(C_{a})$ is the maximum element of set $C_{a}$. We are ready to define the following polynomials.
\begin{definition}
For every $a\in A$ and every integer $k$ such that $0\leq k\leq s-1$, let
$$e_{a,k}(x)=\sum_{j=0}^{t-1}\gamma^{q^{j+k}}x^{q^{j}a},$$
where $\gamma$ is a fixed normal element of $\mathbb{F}_{q^{s}}/\mathbb{F}_{q}$, i.e., $\gamma,\gamma^{q},\cdots,\gamma^{q^{s-1}}$ is an $\mathbb{F}_{q}$-basis of $\mathbb{F}_{q^{s}}$.
\end{definition}
Let $B:=\{e_{a,k}(x)|a\in A,0\leq k\leq s-1\}$ be the set of all polynomials just defined and $U_{m}$ be the subgroup of the $m$-th roots of unity in $\mathbb{F}_{q^{t}}^{*}$. Then we have the following result.
\begin{lemma}\label{lemma1}
The polynomials $e_{a,k}(x)$ have the following properties:
\begin{enumerate}
  \item[(i)] For $a\in A$ and $0\leq k\leq s-1$, the polynomial $e_{a,k}(x)$ has coefficients in $\mathbb{F}_{q^{s}}$ and degree equal to $a$.
  \item[(ii)] The polynomials $e_{a,k}(x)$ for $a\in A$ and $0\leq k\leq s-1$ are linearly independent over $\mathbb{F}_{q}$.
  \item[(iii)] $e_{a,k}(\beta)\in\mathbb{F}_{q}$ for all $\beta\in U_{m}$.
  \item[(iv)] $e_{a,k}(u)=0$ for all $u\in(\mathbb{F}_{q^{s}}\bigcap U_{m})\bigcup\{0\}$.
  \item[(v)] $|B|=\frac{m-\textup{gcd}(m,q^{s}-1)}{p}$.
\end{enumerate}
\end{lemma}
\begin{proof}
(i). All the assertions are clear.

(ii). If $a$ and $b$ are in different $q$-cyclotomic cosets, then $e_{a,k_{1}}(x)$ and $e_{b,k_{2}}(x)$ have different degrees, they are linearly independent. Since $\gamma,\gamma^{q},\cdots,\gamma^{q^{s-1}}$ are an $\mathbb{F}_{q}$-basis of $\mathbb{F}_{q^{s}}$, $e_{a,k}(x)$ for $0\leq k\leq s-1$ are linearly independent over $\mathbb{F}_{q}$.

(iii). Let $\beta\in U_{m}$, then
 \begin{align*}
(e_{a,k}(\beta))^{q}&=(\sum_{j=0}^{t-1}\gamma^{q^{j+k}}\beta^{q^{j}a})^{q} \\
          &=\sum_{j=0}^{t-1}\gamma^{q^{j+k+1}}\beta^{q^{j+1}a}\\
          &=\sum_{j=1}^{t-1}\gamma^{q^{j+k}}\beta^{q^{j}a}+\gamma^{q^{t+k}}\beta^{q^{t}a}\\
          &=\sum_{j=0}^{t-1}\gamma^{q^{j+k}}\beta^{q^{j}a},
\end{align*}
so $e_{a,k}(\beta)\in\mathbb{F}_{q}$.

(iv). It is clear that $e_{a,k}(0)=0$. Let $u\in\mathbb{F}_{q^{s}}\bigcap U_{m}$, then
$$e_{a,k}(u)=\sum_{j=0}^{t-1}\gamma^{q^{j+k}}u^{q^{j}a}=\frac{t}{s}\sum_{j=0}^{s-1}\gamma^{q^{j+k}}u^{q^{j}a}=p\sum_{j=0}^{s-1}\gamma^{q^{j+k}}u^{q^{j}a}=0.$$

(v). We can verify that $|C_{a}|<t$ if and only if $m|a(q^{s}-1)$, that is $\frac{m}{\textup{gcd}(m,q^{s}-1)}|a$. Since every $q$-cyclotomic coset with $t$ elements corresponds to $s$ polynomials, we have $|B|=\frac{m-\textup{gcd}(m,q^{s}-1)}{p}$.
\end{proof}

In order to give our construction, let
\begin{itemize}
  \item $L\subseteq A$, $S:=\bigcup_{a\in L}C_{a}$, $\mathfrak{A}:=\bigcup_{a\in A}C_{a}$,
  \vspace{5pt}
  \item $B(S):=\{e_{a,k}(x)|a\in S, 0\leq k\leq s-1\}$,
  \vspace{5pt}
  \item $V(S):=\textup{span}_{\mathbb{F}_{q}}(B(S))$,
  \vspace{5pt}
  \item $\{\beta_{1},\beta_{2},\cdots,\beta_{n}\}$ be a complete set of representatives of elements from $U_{m}\backslash\mathbb{F}_{q^{s}}^{*}$ which are pairwise nonconjugate over $\mathbb{F}_{q^{s}}$, where $U_{m}$ is the subgroup of the $m$-th roots of unity in $\mathbb{F}_{q^{t}}^{*}$.
\end{itemize}
It is clear that $n=\frac{m-\textup{gcd}(m,q^{s}-1)}{p}$. Our construction is:
\begin{proposition}\label{prop3}
Let
$$C(S):=\{(f(\beta_{1}),f(\beta_{2}),\cdots,f(\beta_{n}))|f\in V(S)\}.$$
Then $C(S)$ is an $\mathbb{F}_{q}$-linear $[n,k,d]$ code, where $n=\frac{m-\textup{gcd}(m,q^{s}-1)}{p}$, $k=|B(S)|$, $d\geq\lceil \frac{m+1-c}{p}\rceil$ and $c$ is the maximum element of set $S$.
\end{proposition}
\begin{proof}
We only need to show that $d\geq\lceil \frac{m+1-c}{p}\rceil$. On one hand, by Lemma~\ref{lemma1}, $f(u)=0$ for all $u\in(\mathbb{F}_{q^{s}}\bigcap U_{m})\bigcup\{0\}$. On the other hand, if $\beta_{i}$ is a root of $f(x)$, then so are all the $p$ conjugate elements $\beta_{i},\cdots,\beta_{i}^{q^{s(p-1)}}$. Therefore $f(x)$ has at most $\frac{\textup{deg}(f(x))-\textup{gcd}(m,q^{s}-1)-1}{p}$ roots among $\{\beta_{1},\beta_{2},\cdots,\beta_{n}\}$. Hence its Hamming weight is at least $n-\frac{c-\textup{gcd}(m,q^{s}-1)-1}{p}=\frac{m+1-c}{p}$.
\end{proof}

\subsection{New Quantum Codes}
In this section, we will construct some new quantum codes from polynomial codes. First of all, we determine the dual of the code $C(S)$. Let $\overline{C_{a}}=\{m-i|i\in C_{a}\}$ and $\overline{S}=\bigcup_{a\in S}\overline{C_{a}}$, then we have
\begin{proposition}\label{prop1}
The Euclidean dual of $C(S)$ is $C(R)$, where $R=\mathfrak{A}\backslash \overline{S}$.
\end{proposition}
\begin{proof}
Since $\textup{dim}(C(S))+\textup{dim}(C(R))=n$, we only need to show that every codeword in $C(S)$ is orthogonal to all codewords of $C(R)$.

For $a\in S$, $b\in R$ and $0\leq k_{1},k_{2}\leq s-1$, we have
 \begin{align*}
\sum_{i=1}^{n}e_{a,k_{1}}(\beta_{i})e_{b,k_{2}}(\beta_{i})&=\sum_{i=1}^{n}(\sum_{j=0}^{t-1}\gamma^{q^{j+k_{1}}}\beta_{i}^{q^{j}a})(\sum_{l=0}^{t-1}\gamma^{q^{l+k_{2}}}\beta_{i}^{q^{l}b}) \\
&=\sum_{i=1}^{n}\sum_{j=0}^{t-1}\sum_{l=0}^{t-1}\gamma^{q^{j+k_{1}}+q^{l+k_{2}}}\beta_{i}^{q^{j}a+q^{l}b}\\
&=\sum_{j=0}^{t-1}(\sum_{l=0}^{t-1}(\sum_{i=1}^{n}\gamma^{q^{j+k_{1}}+q^{k_{2}}}\beta_{i}^{q^{j}a+b})^{q^{l}})\\
&=\sum_{j=0}^{t-1}(\sum_{l_{1}=0}^{s-1}\sum_{l_{2}=0}^{p-1}(\sum_{i=1}^{n}\gamma^{q^{j+k_{1}}+q^{k_{2}}}\beta_{i}^{q^{j}a+b})^{q^{l_{1}+l_{2}s}})\\
&=\sum_{j=0}^{t-1}\sum_{l_{1}=0}^{s-1}\gamma^{q^{j+k_{1}+l_{1}}+q^{k_{2}+l_{1}}}(\sum_{l_{2}=0}^{p-1}(\sum_{i=1}^{n}\beta_{i}^{q^{j}a+b})^{q^{l_{1}+l_{2}s}}).
\end{align*}
Note that  \[\sum_{l_{2}=0}^{p-1}(\sum_{i=1}^{n}\beta_{i}^{q^{j}a+b})^{q^{l_{1}+l_{2}s}}=\sum_{\beta\in U_{m}}\beta^{q^{j}a+b}-\sum_{\beta\in\mathbb{F}_{q^{s}}\bigcap U_{m}}\beta^{q^{j}a+b}=\begin{cases}0;&\textup{ if }\textup{gcd}(m,q^{s}-1)\nmid q^{j}a+b,\\
-\textup{gcd}(m,q^{s}-1);&\textup{ if }\textup{gcd}(m,q^{s}-1)| q^{j}a+b,\end{cases}\]
we have
 \begin{align*}
\sum_{i=1}^{n}e_{a,k_{1}}(\beta_{i})e_{b,k_{2}}(\beta_{i})&=-\textup{gcd}(m,q^{s}-1)\sum_{\substack{j=0 \\ \textup{gcd}(m,q^{s}-1)| q^{j}a+b}}^{t-1}\sum_{l_{1}=0}^{s-1}\gamma^{q^{j+k_{1}+l_{1}}+q^{k_{2}+l_{1}}}\\
&=-\textup{gcd}(m,q^{s}-1)p\sum_{\substack{j=0 \\ \textup{gcd}(m,q^{s}-1)| q^{j}a+b}}^{s-1}\sum_{l_{1}=0}^{s-1}\gamma^{q^{j+k_{1}+l_{1}}+q^{k_{2}+l_{1}}}\\
&=0.
\end{align*}
Since every codeword of $C(S)$ $(C(R))$ is an $\mathbb{F}_{q}$-linear combination of $(e_{a,k_{1}}(\beta_{i}))_{i}$ $((e_{b,k_{2}}(\beta_{i}))_{i},\textup{ respectively})$, where $a\in S$, $b\in R$ and $0\leq k_{1},k_{2}\leq s-1$. We conclude that the Euclidean dual of $C(S)$ is $C(R)$, where $R=\mathfrak{A}\backslash \overline{S}$.
\end{proof}

In order to apply our result to quantum codes, we want to discuss the Hermitian dual of $C(S)$ as well.

\begin{proposition}\label{propo2}
Let $q=l^{2}$, then the Hermitian dual of $C(S)$ is $C(R)$, where $R=\mathfrak{A}\backslash \overline{lS}$, $lS=\{ls|s\in S\}$.
\end{proposition}
\begin{proof}
It is clear that the Hermitian dual of $C(S)$ is the Euclidean dual of $C(lS)$. Then the desired result follows from Proposition~\ref{prop1}.
\end{proof}

Now we state our main result.
\begin{theorem}\label{mainthm}
Let $q=p^{2e}$ be a prime power, where $p$ is a prime number and $e$ is a positive integer. If there exist an integer $m$ and a finite set $S$ satisfying the following conditions:
\begin{enumerate}
  \item $\textup{gcd}(q,m)=1$ and $\textup{ord}_{m}(q)=p^{b}$, where $b\geq1$ is a positive integer;
  \vspace{5pt}
  \item $L\subseteq A$, $S=\bigcup_{a\in L}C_{a}$, $\mathfrak{A}=\bigcup_{a\in A}C_{a}$, $S\bigcup \overline{p^{e}S}\supseteq\mathfrak{A}$, where $C_{a}$ is $q$-cyclotomic coset modulo $m$ and $A=\{\textup{max}(C_{a})|0\leq a\leq m-1, |C_{a}|=p^{b}\}$;
\end{enumerate}
 then there exists a $p^{e}$-ary quantum code $[[n,k,d]]$, where $n=\frac{m-\textup{gcd}(m,q^{p^{b-1}}-1)}{p}$, $k=\frac{2|S|-m+\textup{gcd}(m,q^{p^{b-1}}-1)}{p}$ and $d\geq\lceil\frac{m+1-c}{p}\rceil$, where $c$ is the maximum element of set $S$.
\end{theorem}
\begin{proof}
By Proposition~\ref{propo2}, we have $C(\mathfrak{A}\backslash\overline{p^{e}S})=C(S)^{\bot H}$. If $S\bigcup \overline{p^{e}S}\supseteq\mathfrak{A}$, then $C(S)^{\bot H}\subseteq C(S)$. Applying Theorem~\ref{thm1} and Proposition~\ref{prop3}, the result follows.
\end{proof}

Table~\ref{Nqc} lists some quantum codes obtained from Theorem~\ref{mainthm}, where $\textup{max}(S)$ is the maximum element of set $S$. In order to do comparison in Table~\ref{compa}, we use the propagation rule (Theorem~\ref{propa}) to obtain some of the codes with lengths that are listed in the table online \cite{E12}. Tables~\ref{compa2} and~\ref{compa} show that our quantum codes have larger minimum distance (larger dimension) than the previous quantum codes available when they have the same length and dimension (minimum distance, respectively).

\begin{table}[h]
\begin{center}
\caption{New Quantum Codes}
\begin{tabular}{|c|c|c|c||c|c|c|c|}
\hline
$q$  &  $m$  & $\textup{max}(S)$ &  quantum codes & $q$  &  $m$ & $\textup{max}(S)$ &  quantum codes\\ \hline
$4$  &  $15$  & $9$ & $[[6,0,\geq4]]_{2}$ & $4$  &  $15$  & $13$ & $[[6,4,\geq2]]_{2}$\\ \hline
$4$  &  $255$  & $226$& $[[120,40,\geq15]]_{2}$ & $4$  &  $255$  & $229$& $[[120,48,\geq14]]_{2}$\\ \hline
$4$  &  $255$  & $230$& $[[120,52,\geq13]]_{2}$ & $4$  &  $255$  & $233$& $[[120,60,\geq12]]_{2}$\\ \hline
$4$  &  $255$  & $234$& $[[120,64,\geq11]]_{2}$ & $4$  &  $255$  & $237$& $[[120,72,\geq10]]_{2}$\\ \hline
$4$  &  $255$  & $241$& $[[120,80,\geq8]]_{2}$ & $4$  &  $255$  & $242$& $[[120,84,\geq7]]_{2}$\\ \hline
$4$  &  $255$  & $245$& $[[120,92,\geq6]]_{2}$ & $4$  &  $255$  & $246$& $[[120,96,\geq5]]_{2}$\\ \hline
$4$  &  $255$  & $249$& $[[120,104,\geq4]]_{2}$ & $4$  &  $255$  & $250$& $[[120,108,\geq3]]_{2}$\\ \hline
$4$  &  $255$  & $253$& $[[120,116,\geq2]]_{2}$ & $9$  &  $104$  & $90$& $[[32,12,\geq5]]_{3}$\\ \hline
$9$  &  $104$  & $94$& $[[32,16,\geq4]]_{3}$ & $9$  &  $104$  & $98$& $[[32,22,\geq3]]_{3}$\\ \hline
$9$  &  $104$  & $101$& $[[32,28,\geq2]]_{3}$ & $9$  &  $728$  & $704$& $[[240,198,\geq9]]_{3}$\\ \hline
$9$  &  $728$  & $707$& $[[240,204,\geq8]]_{3}$ & $9$  &  $728$  & $709$& $[[240,208,\geq7]]_{3}$\\ \hline
$9$  &  $728$  & $713$& $[[240,214,\geq6]]_{3}$ & $9$  &  $728$  & $716$& $[[240,220,\geq5]]_{3}$\\ \hline
$9$  &  $728$  & $718$& $[[240,224,\geq4]]_{3}$ & $9$  &  $728$  & $722$& $[[240,230,\geq3]]_{3}$\\ \hline
$9$  &  $728$  &  $725$&$[[240,236,\geq2]]_{3}$ & $16$  &  $85$  & $67$& $[[40,12,\geq10]]_{4}$\\ \hline
$16$  &  $85$  & $71$& $[[40,16,\geq8]]_{4}$ & $16$  &  $85$  & $73$& $[[40,20,\geq7]]_{4}$\\ \hline
$16$  &  $85$  & $75$& $[[40,22,\geq6]]_{4}$ & $16$  &  $85$  & $77$& $[[40,26,\geq5]]_{4}$\\ \hline
$16$  &  $85$  & $79$& $[[40,30,\geq4]]_{4}$ & $16$  &  $85$  & $81$& $[[40,34,\geq3]]_{4}$\\ \hline
$16$  &  $85$  & $83$& $[[40,38,\geq2]]_{4}$& $16$  &  $255$  & $203$& $[[120,36,\geq27]]_{4}$\\ \hline
$16$  &  $255$  & $209$& $[[120,40,\geq24]]_{4}$ & $16$  &  $255$  & $211$& $[[120,44,\geq23]]_{4}$\\ \hline
$16$  &  $255$  & $213$& $[[120,48,\geq22]]_{4}$ & $16$  &  $255$  & $215$& $[[120,52,\geq21]]_{4}$\\ \hline
$16$  &  $255$  & $217$& $[[120,56,\geq20]]_{4}$ & $16$  &  $255$  & $219$& $[[120,60,\geq19]]_{4}$\\ \hline
$16$  &  $255$  & $220$& $[[120,62,\geq18]]_{4}$ & $16$  &  $255$  & $225$& $[[120,66,\geq16]]_{4}$\\ \hline
$16$  &  $255$  &  $227$&$[[120,70,\geq15]]_{4}$ & $16$  &  $255$  & $229$& $[[120,74,\geq14]]_{4}$\\ \hline
$16$  &  $255$  & $231$& $[[120,78,\geq13]]_{4}$ & $16$  &  $255$  & $233$& $[[120,82,\geq12]]_{4}$\\ \hline
$16$  &  $255$  & $235$& $[[120,86,\geq11]]_{4}$ & $16$  &  $255$  & $237$& $[[120,90,\geq10]]_{4}$\\ \hline
$16$  &  $255$  & $241$& $[[120,94,\geq8]]_{4}$ & $16$  &  $255$  & $243$& $[[120,98,\geq7]]_{4}$\\ \hline
$16$  &  $255$  & $245$& $[[120,102,\geq6]]_{4}$ & $16$  &  $255$  & $247$& $[[120,106,\geq5]]_{4}$\\ \hline
$16$  &  $255$  & $249$& $[[120,110,\geq4]]_{4}$ & $16$  &  $255$  & $251$& $[[120,114,\geq3]]_{4}$\\ \hline
$16$  &  $255$  & $253$& $[[120,118,\geq2]]_{4}$ &   &    &  & \\ \hline
\end{tabular}
\label{Nqc}
\end{center}
\end{table}

\begin{table}[h]
\begin{center}
\caption{Quantum Codes Comparison}
\begin{tabular}{|c|c|}
\hline
quantum codes from Table~\ref{Nqc}  &  quantum codes from \cite{G07}\\ \hline
$[[120,40,\geq15]]_{2}$ & $[[120,40,14]]_{2}$ \\ \hline
$[[120,48,\geq14]]_{2}$ & $[[120,48,13]]_{2}$ \\ \hline
$[[120,52,\geq13]]_{2}$ & $[[120,52,12]]_{2}$ \\ \hline
$[[120,60,\geq12]]_{2}$ & $[[120,60,11]]_{2}$ \\ \hline
$[[120,64,\geq11]]_{2}$ & $[[120,64,10]]_{2}$ \\ \hline
$[[120,72,\geq10]]_{2}$ & $[[120,72,9]]_{2}$ \\ \hline
\end{tabular}
\label{compa2}
\end{center}
\end{table}

\begin{table}[h]
\begin{center}
\caption{Quantum Codes Comparison}
\begin{tabular}{|c|c|c|}
\hline
quantum codes from Table~\ref{Nqc}  & using propagation rule  &  quantum codes from \cite{E12}\\ \hline
$[[240,220,\geq5]]_{3}$ &  $[[238,220,\geq3]]_{3}$  &  $[[238,216,3]]_{3}$ \\ \hline
$[[40,16,\geq8]]_{4}$ &                           &  $[[40,2,8]]_{4}$ \\ \hline
$[[40,20,\geq7]]_{4}$ &                           &  $[[40,8,7]]_{4}$ \\ \hline
$[[40,22,\geq6]]_{4}$ &                           &  $[[40,14,6]]_{4}$ \\ \hline
$[[40,26,\geq5]]_{4}$ &                           &  $[[40,20,5]]_{4}$ \\ \hline
$[[40,30,\geq4]]_{4}$ &                           &  $[[40,26,4]]_{4}$ \\ \hline
$[[40,34,\geq3]]_{4}$ &                           &  $[[40,32,3]]_{4}$ \\ \hline
$[[120,52,\geq21]]_{4}$ &  $[[117,52,\geq18]]_{4}$  &  $[[117,49,14]]_{4}$ \\ \hline
$[[120,56,\geq20]]_{4}$ &  $[[117,56,\geq17]]_{4}$  &  $[[117,49,14]]_{4}$ \\ \hline
$[[120,60,\geq19]]_{4}$ &  $[[117,60,\geq16]]_{4}$  &  $[[117,49,14]]_{4}$ \\ \hline
$[[120,62,\geq18]]_{4}$ &  $[[117,62,\geq15]]_{4}$  &  $[[117,49,14]]_{4}$ \\ \hline
\end{tabular}
\label{compa}
\end{center}
\end{table}

\section{New Quantum Synchronizable Codes from Duadic Codes}\label{quantumsyn}
In this section we study a special class of cyclic codes to give a new family of quantum synchronizable codes.
\subsection{Duadic Codes}
In this subsection, we recall the definition and basic properties of duadic codes of length $n$ over a finite field $\mathbb{F}_{q}$ such that $\textup{gcd}(n,q)=1$.

Let $S_{0}$, $S_{1}$ be the defining sets of two cyclic codes of length $n$ over $\mathbb{F}_{q}$ such that
\begin{enumerate}
  \item $S_{0}\bigcap S_{1}=\emptyset,$
  \item $S_{0}\bigcup S_{1}=S=\{1,2,...,n-1\},$ and
  \item $aS_{i}\pmod{n}=S_{i+1\pmod{2}}$ for some $a$ coprime to $n$.
\end{enumerate}
In particular, each $S_{i}$ is a union of $q$-ary cyclotomic cosets modulo $n$. Then $|S_{0}|=|S_{1}|$, we have $|S_{i}|=\frac{n-1}{2}$. Hence $n$ must be odd.

Let $\omega$ be a primitive $n$-th root of unity over a field $\mathbb{F}_{q}$. For $i\in\{0,1\}$, the odd-like duadic code $D_{i}$ is a cyclic code of length $n$ over $\mathbb{F}_{q}$ with defining set $S_{i}$ and generator polynomial
$$g_{i}(x)=\prod_{j\in S_{i}}(x-\omega^{j}).$$
The even-like duadic code $C_{i}$ is defined as  a cyclic code with defining set $S_{i}\bigcup\{0\}$ and generator polynomial $(x-1)g_{i}(x)$. Then the dimension of $D_{i}$ is $(n+1)/2$ and dimension of $C_{i}$ is $(n-1)/2$ respectively. Obviously $C_{i}\subseteq D_{i}$.
\begin{lemma}\rm{\cite{HP03}}\label{dud1}
Let $C_{i}$ and $D_{i}$ be the even-like and odd-like duadic codes of length $n$ over $\mathbb{F}_{q}$, where $i\in\{0,1\}$. Then
\begin{enumerate}
  \item $C_{i}^{\bot}=D_{i}$ if and only if $-S_{i}\equiv S_{i+1\pmod{2}}\pmod{n}$.
  \item $C_{i}^{\bot}=D_{i+1\pmod{2}}$ if and only if $-S_{i}\equiv S_{i}\pmod{n}$.
\end{enumerate}
\end{lemma}

The following well known fact gives a lower bound for the minimum distance of odd-like duadic codes.

\begin{lemma}\rm{\cite{HP03}}\label{dud2}
Let $D_{0}$ and $D_{1}$ be a pair of odd-like duadic codes of length $n$ over $\mathbb{F}_{q}$. If $-S_{i}\equiv S_{i+1\pmod{2}}\pmod{n}$, then their minimum distances in both codes are the same, say $d_{0}$. We also have $d_{0}^{2}-d_{0}+1\geq n$.
\end{lemma}

\subsection{New Quantum Synchronizable Codes}
We now provide a construction for quantum synchronizable codes designed from duadic codes. We first recall the following lemma, which gives a connection between cyclotomic cosets and irreducible polynomials.
\begin{lemma}\rm{\cite{HP03}}\label{cyc}
Let $\omega$ be a primitive $n$-th root of unity over a field containing $\mathbb{F}_{q}$, where $\textup{gcd}(q,n)=1$. Then the minimal polynomial of $\omega^{i}$ with respect to $\mathbb{F}_{q}$ is
$$M^{(i)}(x)=\prod_{j\in C_{i,n}}(x-\omega^{j}),$$
where $C_{i,n}$ is the unique $q$-cyclotomic coset modulo $n$ containing $i$.
\end{lemma}
 Now we assume that $p$ is an odd prime and $p\equiv-1\pmod{8}$. Let $m,n$ be positive integers, $m\leq n$, $g$ be a primitive root modulo $p^{n}$, and $\omega$ be a primitive $p^{n}$-th root of unity over a field containing $\mathbb{F}_{2}$. Define
$$g_{m}=g\pmod{p^{m}},\ \omega_{m}=\omega^{p^{n-m}}.$$
Then $g_{m}$ is a primitive root modulo $p^{m}$ and $\omega_{m}$ is a primitive $p^{m}$-th root of unity.

Set
$$S_{m0}=\langle g_{m}^{2}\rangle,\ S_{m1}=g_{m}S_{m0},$$
where $\langle g_{m}^{2}\rangle$ denotes the subgroup generated by $g_{m}^{2}$ of $\mathbb{Z}_{p^{m}}^{*}$. It is obvious that
$$S_{m0}\bigcup S_{m1}=\mathbb{Z}_{p^{m}}^{*},\ S_{m0}\bigcap S_{m1}=\emptyset.$$
We also define polynomials corresponding to $S_{mj},\ j=0,1$:
$$d_{mj}(x)=\prod_{i\in S_{mj}}(x-\omega_{m}^{i}),\ j=0,1.$$
We can also verify that
$$x^{p^{n}}-1=(x-1)(\prod_{m=1}^{n}d_{m0}(x))(\prod_{m=1}^{n}d_{m1}(x)).$$
Since $p\equiv-1\pmod{8}$, then $2\in S_{m0}$, whence $d_{m0}(x)\in\mathbb{F}_{2}[x]$ from Lemma~\ref{cyc}. Let
$$g_{m}(x)=d_{10}(x)d_{20}(x)\cdots d_{m0}(x),$$
and $C_{m}$ denote the cyclic code over $\mathbb{F}_{2}$ of length $p^{n}$ generated by the polynomial $g_{m}(x)$.

\begin{lemma}
$\textup{ord}(d_{n0}(x))=p^{n}$.
\end{lemma}
\begin{proof}
Since all the roots of $d_{n0}(x)$ are $p^{n}$-th root of unity, and $d_{n0}(\omega_{n})=0$, where $\omega_{n}$ is a primitive $p^{n}$-th root of unity. Then by Lemmas~\ref{ord1} and \ref{ord2}, $\textup{ord}(d_{n0}(x))=p^{n}$.
\end{proof}

\begin{theorem}
Let $n>1$ be a positive integer, $p$ be an odd prime and $p\equiv-1\pmod{8}$. Then for pair $a_{l},\ a_{r}$ of nonnegative integers such that $a_{l}+a_{r}<p^{n}$, there exists a quantum synchronizable $(a_{l},a_{r})-[[p^{n}+a_{l}+a_{r},1]]_{2}$ code that corrects at least up to $\lfloor\frac{d_{n}-1}{2}\rfloor$ phase errors, where $d_{n}^{2}-d_{n}+1\geq p^{n}$.
\end{theorem}
\begin{proof}
Take the code $C_{n}$ and $C_{m}$, where $m<n$, then $C_{n}\subseteq C_{m}$. By Lemma~\ref{dud2}, $C_{n}$ is an odd-like duadic code with parameters $[p^{n},\frac{p^{n}+1}{2},d_{n}]$, where $d_{n}^{2}-d_{n}+1\geq p^{n}$. Since $p\equiv-1\pmod{8}$, we have $S_{n0}=-S_{n1}$. By Lemma~\ref{dud1}, $C_{n}\supseteq C_{n}^{\bot}$. Applying Theorem~\ref{syn}, the assertion follows.
\end{proof}

In order to get more results, we discuss the factorization of $d_{n0}(x)$.
\begin{lemma}\label{cyc1}
Let $n>1$ be a positive integer, $p$ be an odd prime and $p\equiv-1\pmod{8}$. If $\textup{ord}_{p^{n}}(2)=t$, then $d_{n0}(x)$ can be factorized into $\frac{p^{n-1}(p-1)}{2t}$ irreducible polynomials of degree $t$ over $\mathbb{F}_{2}$.
\end{lemma}
\begin{proof}
Let $S_{i,p^{n}}=\{i\cdot 2^{j}\pmod{p^{n}}|j=0,1,\ldots\}$ denoting the $2$-cyclotomic coset modulo $p^{n}$ containing $i$. If $\textup{ord}_{p^{n}}(2)=t$, then $t$ is the minimal integer such that $p^{n}|2^{t}-1$, so $|S_{i,p^{n}}|=t$ for any $\textup{gcd}(i,p^{n})=1$. Since all the roots of $d_{n0}(x)$ are with the form $\omega^{j}$, where $\textup{gcd}(j,p^{n})=1$. Then by Lemma~\ref{cyc}, $d_{n0}(x)$ can be factorized into $\frac{p^{n-1}(p-1)}{2t}$ irreducible polynomials of degree $t$ over $\mathbb{F}_{2}$.
\end{proof}

\begin{theorem}\label{thm2}
Let $n\geq1$ be a positive integer, $p$ be an odd prime and $p\equiv-1\pmod{8}$. If $\textup{ord}_{p^{m}}(2)=t_{m},\ m=1,2,\cdots n$. Then for nonnegative integers $a_{l},\ a_{r}$, $u_{m}\ (m=1,2,\cdots,n)$ such that $a_{l}+a_{r}<p^{n}$, $1\leq u_{n}\leq\frac{p^{n-1}(p-1)}{2t_{n}}, 0\leq u_{l}\leq\frac{p^{l-1}(p-1)}{2t_{l}}\textup{ for }l=1,2,\cdots,n-1,\textup{ and } \sum_{l=1}^{n-1}u_{l}\geq1$, there exists a quantum synchronizable $(a_{l},a_{r})-[[p^{n}+a_{l}+a_{r},p^{n}-2\sum_{i=1}^{n}u_{i}t_{i}]]_{2}$ code.
\end{theorem}
\begin{proof}
By Lemma~\ref{cyc1}, $d_{m0}(x)=\prod_{j=1}^{\frac{p^{m-1}(p-1)}{2t_{m}}}h_{mj}(x)$, where $h_{mj}(x),\ j=1,\cdots,\frac{p^{n-1}(p-1)}{2t_{m}},\ m=1,2,\cdots,n$ are irreducible polynomials over $\mathbb{F}_{2}$.

Then let
$$f_{1}(x)=\prod_{i=1}^{n}\prod_{j=1}^{u_{i}}h_{ij}(x),\ 1\leq u_{n}\leq\frac{p^{n-1}(p-1)}{2t_{n}}, 0\leq u_{l}\leq\frac{p^{l-1}(p-1)}{2t_{l}}\textup{ for }l=1,2,\cdots,n-1,\textup{ and } \sum_{l=1}^{n-1}u_{l}\geq1,$$
$$f_{2}(x)=\prod_{i=1}^{n}\prod_{j=1}^{v_{i}}h_{ij}(x),\ v_{n}<u_{n}, v_{l}\leq u_{l}\textup{ for }l=1,2,\cdots,n-1,\textup{ and } \sum_{l=1}^{n}v_{l}\geq1,$$
and $D_{1}\ (D_{2})$ denote the cyclic code of length $p^{n}$ generated by the polynomial $f_{1}(x)$ ($f_{2}(x)$, respectively).

Then applying Theorem~\ref{syn}, for nonnegative integers $a_{l},\ a_{r}$ and $u_{m}\ (m=1,2,\cdots,n)$ such that $a_{l}+a_{r}<p^{n}$, $1\leq u_{n}\leq\frac{p^{n-1}(p-1)}{2t_{n}}, 0\leq u_{l}\leq\frac{p^{l-1}(p-1)}{2t_{l}}\textup{ for }l=1,2,\cdots,n-1,\textup{ and } \sum_{l=1}^{n-1}u_{l}\geq1$, we obtain a quantum synchronizable $(a_{l},a_{r})-[[p^{n},p^{n}-2\sum_{i=1}^{n}u_{i}t_{i}]]_{2}$ code.
\end{proof}
\begin{remark}
It is easy to see that Theorem~\ref{thm2} is a generalization of the result in \cite{XYF14}.
\end{remark}

\begin{example}
Let $p=31$, $n=2$. It can be verified that $g=3$ is a primitive root modulo $31^{2}$. Take
$$S_{20}=\{9^{i}\pmod{961}|i\in\mathbb{N}\},\ S_{10}=\{9^{i}\pmod{31}|i\in\mathbb{N}\}.$$
Then $|S_{20}|=465$ and $|S_{10}|=15$. We can get that $\textup{ord}_{961}(2)=155$ and $\textup{ord}_{31}(2)=5$, then $S_{20}$ ($S_{10}$) is the union of $3$ cyclotomic cosets as follows:
$$S_{20}=T_{21}\bigcup T_{22}\bigcup T_{23},\ S_{10}=T_{11}\bigcup T_{12}\bigcup T_{13},$$
where $T_{21}=\{1, 2, 4, 8, 16, 32, 33, 35,\cdots\}$, $T_{22}=\{5, 9, 10, 18, 20, 36, 40, 41,\cdots\}$, $T_{23}=\{7, 14, 19, 25, 28, 38, 45, 50,\cdots\}$,
$T_{11}=\{1, 2, 4, 8, 16\}$, $T_{12}=\{5, 9, 10, 18, 20\}$ and $T_{13}=\{7, 14, 19, 25, 28\}$.
Let $\omega$ be a primitive $961$-th root of unity over a field containing $\mathbb{F}_{2}$. Set
$$h_{2i}(x)=\prod_{j\in T_{2i}}(x-\omega^{j}),\ h_{1i}(x)=\prod_{j\in T_{1i}}(x-\omega^{31j}),\ i=1,2,3.$$
Then $h_{ij}(x)\in\mathbb{F}_{2}[x]$ for $i=1,2,\ j=1,2,3.$
Applying Theorem~\ref{thm2}, we obtain the quantum synchronizable $(a_{l},a_{r})-[[961+a_{l}+a_{r},k]]_{2}$ codes listed in table~\ref{Quantum}, where $a_{l}+a_{r}<961$.
\begin{table}[h]
\begin{center}
\caption{Quantum Synchronizable $(a_{l},a_{r})-[[961+a_{l}+a_{r},k]]_{2}$ Codes, where $a_{l}+a_{r}<961$}
\begin{tabular}{|c|c|c|c|}
\hline
$f_{1}(x)$  &  $f_{2}(x)$  &  $k_{1}$  & $k$  \\ \hline
$\Pi_{i=1}^{3}h_{1i}(x)\Pi_{j=1}^{3}h_{2j}(x)$& $\Pi_{i=1}^{3}h_{1i}(x)\Pi_{j=1}^{2}h_{2j}(x)$ &   481    &   1  \\ \hline
$\Pi_{i=1}^{3}h_{1i}(x)\Pi_{j=1}^{2}h_{2j}(x)$& $\Pi_{i=1}^{3}h_{1i}(x)h_{21}(x)$ &   636    &   311  \\ \hline
$\Pi_{i=1}^{2}h_{1i}(x)\Pi_{j=1}^{3}h_{2j}(x)$& $\Pi_{i=1}^{2}h_{1i}(x)\Pi_{j=1}^{2}h_{2j}(x)$ &   486    &   11    \\ \hline
$\Pi_{i=1}^{2}h_{1i}(x)\Pi_{j=1}^{2}h_{2j}(x)$& $\Pi_{i=1}^{2}h_{1i}(x)h_{21}(x)$ &   641    &   321    \\ \hline
$h_{11}(x)\Pi_{j=1}^{3}h_{2j}(x)$& $h_{11}(x)\Pi_{j=1}^{2}h_{2j}(x)$ &   491    &   21    \\ \hline
$\Pi_{i=1}^{3}h_{1i}(x)h_{21}(x)$& $\Pi_{i=1}^{3}h_{1i}(x)$  &   791    &   621    \\ \hline
$\Pi_{j=1}^{3}h_{2j}(x)$& $\Pi_{j=1}^{2}h_{2j}(x)$  &   496    &   31   \\ \hline
$h_{11}(x)\Pi_{j=1}^{2}h_{2j}(x)$& $h_{11}(x)h_{21}(x)$ &   646    &   331    \\ \hline
$\Pi_{i=1}^{2}h_{1i}(x)h_{21}(x)$& $\Pi_{i=1}^{2}h_{1i}(x)$ &   796    &   631     \\ \hline
$\Pi_{j=1}^{2}h_{2j}(x)$& $h_{21}(x)$ &   651    &   341     \\ \hline
$h_{11}(x)h_{21}(x)$& $h_{11}(x)$ &   801    &   641     \\ \hline
\end{tabular}
\label{Quantum}
\end{center}
\end{table}
\end{example}

\end{document}